\documentclass[11pt,a4paper]{article}

\usepackage[margin=1in]{geometry} 
\usepackage{graphics,graphicx,color}
\usepackage{epsfig}
\usepackage{amsthm,amssymb,amsmath,amsfonts}
\usepackage{mathrsfs}
\usepackage{algorithm,algorithmic}
\usepackage{subfigure}
\usepackage{hyperref}	
\usepackage{tikz,pgf} 
\usepackage{todonotes}
\usetikzlibrary{positioning,shapes,shadows,arrows} 

\newtheorem{lemma}{Lemma}
\newtheorem{theorem}{Theorem}
\newtheorem{corollary}{Corollary}

\title{Parameterized Integer Quadratic Programming: Variables and Coefficients}

\author{Daniel Lokshtanov\thanks{Department of Informatics, University of Bergen, Norway, e-mail: \texttt{daniello@ii.uib.no}}}

\begin{document}

\maketitle

\begin{abstract}
In the {\sc Integer Quadratic Programming} problem input is an $n \times n$ integer matrix $Q$, an $m \times n$ integer matrix $A$ and an $m$-dimensional integer vector $b$. The task is to find a vector $x \in \mathbb{Z}^n$ minimizing $x^TQx$, subject to $Ax \leq b$. We give a fixed parameter tractable algorithm for {\sc Integer Quadratic Programming} parameterized by $n + \alpha$. Here $\alpha$ is the largest absolute value of an entry of $Q$ and $A$. As an application of our main result we show that {\sc Optimal Linear Arrangement} is fixed parameter tractable parameterized by the size of the smallest vertex cover of the input graph. This resolves an open problem from the recent monograph by Downey and Fellows.
\end{abstract}
 
 
\section{Introduction}
While {\sc Linear Programming} is famously polynomial time solvable~\cite{khachiyan1980polynomial}, most generalizations are not. In particular, requiring the variables to take integer values gives us the {\sc Integer Linear Programming} problem, which is easily seen to be NP-hard. On the other hand, integer linear programs (ILPs) with few variables can be solved efficiently. The celebrated algorithm of Lenstra~\cite{Lenstra83} solves ILPs with $n$ variables in time $f(n)L^{O(1)}$ where $f$ is a (doubly exponential) function depending only on $n$ and $L$ is the total number of bits required to encode the input integer linear program. In terms of parameterized complexity this means that  {\sc Integer Linear Programming} is {\em fixed parameter tractable} (FPT) when parameterized by the number $n$ of variables to the input ILP. In parameterized complexity input instances come with a {\em parameter} $k$, and an algorithm is called fixed parameter tractable if it solves instances of size $L$ with parameter $k$ in time $f(k)L^{O(1)}$ for some function $f$ depending only on $k$. For an introduction to parameterized complexity we refer to the recent monograph of Downey and Fellows~\cite{DowneyF13}, as well as the textbook by Cygan et al.~\cite{the-awesome-book}.

Following the algorithm of Lenstra~\cite{Lenstra83} there has been a significant amount of research into parameterized algorithms for  {\sc Integer Linear Programming}, as well as generalizations of the problem to (quasi) convex optimization. Highlights include the algorithms for {\sc Integer Linear Programming} with improved dependence on $n$ by Kannan~\cite{Kannan}, Clarkson~\cite{Clarkson95} and Frank and Tardos~\cite{FrankT87} and generalizations to $N$-fold integer programming due to Hemmecke et al.~\cite{HemmeckeOR13}, see also the book by Onn~\cite{onn2010nonlinear}.
Heinz~\cite{Heinz05} generalized the FPT algorithm of Lenstra to quasi-convex polynomial optimization. More concretely, the algorithm of Heinz finds an integer assignment to variables $x_1, \ldots, x_n$ minimizing $f(x_1, \ldots x_n)$ subject to the constraints $g_i(x_1, \ldots, x_n) \leq 0$ for $1 \leq  i \leq m$, where $f$ and $g_1, \ldots, g_m$ are quasi-convex polynomials of degree $d \geq 2$. Here a function $f : \mathbb{R}^n \rightarrow \mathbb{R}$ is {\em quasi-convex} if for every real $\lambda$ the set $\{x \in \mathbb{R}^n : f(x) \leq \lambda\}$ is convex. The algorithm has running time $L^{O(1)}n^{O(dn)}2^{O(n^3)}$, that is, it is fixed parameter tractable in the dimension $n$ and the degree $d$ of the input polynomials. Khachiyan and Porkolab~\cite{KhachiyanP00} gave an even more general algorithm that covers the case of minimization of convex polynomials over the integer points in convex semialgebraic sets given by arbitrary (not necessarily quasi-convex) polynomials, see~\cite{koppe2012complexity} for more details.


On the other hand, generalizations of {\sc Integer Linear Programming} to optimization of possibly non-convex functions over possibly non-convex domains quickly become computationally intractable in the strongest possible sense. In fact, solving a system of quadratic equations over $232$ integer valued variables is undecidable~\cite{koppe2012complexity}, and the same holds for finding an integer root of a single multi-variate polynomial of degree $4$~\cite{koppe2012complexity}. 
Nevertheless, there are interesting special cases of non-convex (integer) mathematical programming for which algorithms are known to exist, and it is an under-explored research direction to investigate the parameterized complexity of these problems.
%
Perhaps the simplest such generalization is the {\sc Integer Quadratic Programming} problem. Here the input is a $n \times n$ integer matrix $Q$, an $m \times n$ integer matrix $A$ and an $m$-dimensional integer vector $b$. The task is to find a vector $x \in \mathbb{Z}^n$ minimizing $x^TQx$, subject to $Ax \leq b$. Thus, in this problem, the domain is convex, but the objective function might not be.
%
It is a major open problem whether there exists a polynomial time algorithm for  {\sc Integer Quadratic Programming}  with a constant number of variables. 
Indeed, until quite recently the problem was not even known to be in NP~\cite{PiaDM16}, and the first polynomial time algorithm for {\sc Integer Quadratic Programming} in {\em two} variables was given by Del Pia and Weismantel~\cite{PiaW14} in 2014. 



   In this paper we take a more modest approach to {\sc Integer Quadratic Programming}, and consider the problem when parameterized by the number $n$ of variables and the largest absolute value $\alpha$ of the entries in the matrices $Q$ and $A$. Our main result is an algorithm for  {\sc Integer Quadratic Programming} with running time $f(n,\alpha)L^{O(1)}$, demonstrating that the problem is fixed parameter tractable when parameterized by the number of variables and the largest coefficient appearing in the objective function and in the constraints. 

On one hand {\sc Integer Quadratic Programming} is a more general problem than  {\sc Integer Linear Programming}. On the other hand the parameterization by variables and coefficients is a much stronger parameterization than parameterizing just by the number $n$ of variables. By making the largest entry $\alpha$ of $Q$ and $A$ a parameter we allow the running time of our algorithms to depend in arbitrary ways on essentially all of the input. The only reason that designing an FPT algorithm for this parameterization is non-trivial is that the entries in the vector $b$ may be arbitrarily large compared to the parameters $n$ and $\alpha$. This makes the number of possible assignments to the variables much too large to enumerate all assignments by brute force. Indeed, despite being quite restricted our algorithm for {\sc Integer Quadratic Programming} allows us to show fixed parameter tractability of a problem whose parameterized complexity was unknown prior to this work. More concretely we use the new algorithm for  {\sc Integer Quadratic Programming}  to prove that  {\sc Optimal Linear Arrangement} parameterized by the size of the smallest vertex cover of the input graph is fixed parameter tractable.

In the {\sc Optimal Linear Arrangement}  problem we are given as input an undirected graph $G$ on $n$ vertices. The task is to find a permutation $\sigma : V(G) \rightarrow \{1, \ldots, n\}$ minimizing the {\em cost} of $\sigma$. Here the cost of a permutation $\sigma$ is $val(\sigma,G) = \sum_{uv \in E(G)} |\sigma(u)-\sigma(v)|$. The problem was shown to be NP-complete already in the 70's~\cite{GJ79}, admits a factor $O(\sqrt{\log n}\log\log n)$ approximation algorithm~\cite{CharikarHKR10,FeigeL07}, but no admits no polynomial time approximation scheme, assuming plausible complexity-theoretic assumptions~\cite{AmbuhlMS11}.

We consider {\sc Optimal Linear Arrangement} parameterized by the size of the smallest {\em vertex cover} of the input graph $G$. A {\em vertex cover} of a graph $G$ is a vertex set $C$ such that every edge in $G$ has at least one endpoint in $C$. When {\sc Optimal Linear Arrangement} is parameterized by the vertex cover number of the input graph, an integer parameter $k$ is also given as input together with $G$ and $n$. An FPT algorithm is allowed to run in time $f(k)n^{O(1)}$ and only has to provide an optimal layout $\sigma$ of $G$ if there exists a vertex cover in $G$ of size at most $k$. We remark that one can compute a vertex cover of size $k$, if it exists, in time $O(1.2748^k + n^{O(1)})$~\cite{ChenKX10}. Hence, when designing an algorithm for {\sc Optimal Linear Arrangement} parameterized by vertex cover we may just as well assume that a vertex cover $C$ of $G$ of size at most $k$ is given as input.

The parameterized complexity of {\sc Optimal Linear Arrangement} parameterized by vertex cover was first posed as an open problem by Fellows et al.~\cite{FellowsLMRS08}. Fellows et al.~\cite{FellowsLMRS08} showed that a number of well-studied graph layout problems, such as {\sc Bandwidth} and {\sc Cutwidth} can be shown to be FPT when parameterized by vertex cover, by reducing the problems to  {\sc Integer Linear Programming} parameterized by the number of variables. For the most natural formalization of 
{\sc Optimal Linear Arrangement} as an integer program the objective function is quadratic (and not necessarily convex), and therefore the above approach fails.

Motivated by the lack of progress on this problem, Fellows et al.~\cite{FellowsHRS13} recently showed an FPT approximation scheme for {\sc Optimal Linear Arrangement} parameterized by vertex cover. In partiular they gave an algorithm that given as input a graph $G$ with a vertex cover of size at most $k$ and a rational $\epsilon > 0$, produces in time $f(k,\epsilon)n^{O(1)}$ a layout $\sigma$ with cost at most a factor $(1+\epsilon)$ larger than the optimum. Fellows et al.~\cite{FellowsHRS13} re-state the parameterized complexity of  {\sc Optimal Linear Arrangement} parameterized by vertex cover as an open problem. Finally, the problem was re-stated as an open problem in the recent monograph of Downey and Fellows~\cite{DowneyF13}. Interestingly, Downey and Fellows motivate the study of this problem as follows. 

\smallskip
 ``{\em Our enthusiasm for this concrete problem is based on its connection to {\sc Integer Linear Programming}. The problem above is easily reducible to a restricted form of {\sc Integer Quadratic Programming} which may well be FPT}''.
\smallskip

We give an FPT algorithm for {\sc Optimal Linear Arrangement} parameterized by vertex cover, resolving the open problem of~\cite{DowneyF13,FellowsHRS13,FellowsLMRS08}. Our algorithm for {\sc Optimal Linear Arrangement} works by directly applying the new algorithm for {\sc Integer Quadratic Programming} to the most natural formulation of  {\sc Optimal Linear Arrangement} on graphs with a small vertex cover as an integer quadratic program, confirming the intuition of Downey and Fellows~\cite{DowneyF13}.





\medskip
\noindent
{\bf Preliminaries and notation.}
In Section~\ref{sec:IQPalg} lower case letters denote vectors and scalars, while upper case letters denote matrices. All vectors are column vectors. For an integer $p \geq 2$, the $\ell_p$ norm of an $n$-dimensional vector $v = [v_1, v_2, \ldots, v_n]$ is denoted by $|v|_p$ and is defined to be  $|v|_p = (v_1^p + v_2^p + \ldots + v_n^p)^{1/p}$. The $\ell_1$ norm of $v$ is $|v|_1 = |v_1| + |v_2| + \ldots + |v_n|$, while the  $\ell_\infty$ norm of $v$ is $|v|_\infty = \max(|v_1|, |v_2|, \ldots,|v_n|)$. 


\section{Algorithm for Integer Quadratic Programming}\label{sec:IQPalg}
We consider the following problem, called {\sc Integer Quadratic Programming}. Input consists of an $n \times n$ integer symmetric matrix $Q$, an $m \times n$ integer matrix $A$ and $m$-dimensional integer vector $b$. The task is to find an optimal solution $x^\star$ to the following optimization problem.
\begin{eqnarray}\label{eqn:prob1}
\nonumber \mbox{Minimize } x^TQx \\
\mbox{ subject to:~~} Ax \leq b \\
\nonumber x \in \mathbb{Z}^n.
\end{eqnarray}
A vector  $x \in \mathbb{Z}^n$ that satisfies the constraints $Ax \leq b$ is called a {\em feasible solution} to the IQP~(\ref{eqn:prob1}). Given an input on the form~(\ref{eqn:prob1}) there are three possible scenarios. A possible scenario is that there are no feasible solutions, in which case this is what an algorithm for {\sc Integer Quadratic Programming} should report. Another possibility is that for every integer $\beta$ there exists some feasible solution $x$ such that  $x^TQx \leq \beta$. In that case the algorithm should report that the IQP is {\em unbounded}. Finally, it could be that there exist feasible solutions, and that the minimum value of $x^TQx$ over the set of all feasible $x$ is well defined. This is the most interesting case, and in this case the algorithm should output a feasible $x$ such that $x^TQx$ is minimized.

Note that the requirement that $Q$ is symmetric can easily be avoided by replacing $Q$ by $Q + Q^T$. This operation does not change the set of optimal solutions, since it multiplies the objective function value of every solution by $2$. We will denote by $a_i^T$ the $i$'th {\em row} of the matrix $A$, and by $b_i$ the $i$'th entry of the vector $b$. Thus, $Ax \leq b$ means that $a_i^Tx \leq b_i$ for all $i$. The maximum absolute value of an entry of $A$ and $Q$ is denoted by $\alpha$. Using a pair of inequalities one can encode equality constraints. It is useful to rewrite the IQP~(\ref{eqn:prob1}) to separate out the equality constraints explicitely, obtaining the following equivalent form. 
\begin{eqnarray}\label{eqn:prob}
\nonumber \mbox{Minimize } x^TQx \\
\mbox{ subject to:~~} Ax \leq b \\
\nonumber Cx = d \\
\nonumber x \in \mathbb{Z}^n.
\end{eqnarray}
Here $C$ is an integer matrix and $d$ is an integer vector. If input is given on the form~(\ref{eqn:prob}), then we still use $\alpha$ to denote the maximum value of an entry of $A$ and $Q$. The IQP~(\ref{eqn:prob}) could be generalized by changing the objective function from  $x^TQx$ to  $x^TQx + q^Tx$ for some n-dimensional vector $q$ also given as input. This generalization can be incorporated in the original formulation~(\ref{eqn:prob}) at the cost of introducing a new variable $\hat{x}$, adding the constraint $\hat{x} = 1$ to the system $Cx = d$ and adding $[0,q]$ as the row corresponding to the new variable $\hat{x}$ in $Q$.

We will denote by $\Delta$ the maximum absolute value of the determinant of a square submatrix of $C$. We may assume without loss of generality that the rows of $C$ are linearly independent; otherwise we may in polynomial time either conclude that the IQP has no feasible solutions, or remove one of the equality constraints in the system $Cx = d$ without changing the set of feasible solutions. Thus $C$ has at most $n$ rows. If the IQP~(\ref{eqn:prob}) is obtained from~(\ref{eqn:prob1}) by replacing constraints $a_i^Tx \leq b_i$, $-a_i^Tx \leq -b_i$ with $a_i^Tx = b_i$, the maximum entry of $C$ is also upper bounded by $\alpha$ and then we have $\Delta \leq n! \cdot \alpha^n$. The next simple observation shows that we can in polynomial time reduce the number of constraints to a function of $n$ and $\alpha$.

\begin{lemma}\label{lem:rowReduce}
There is a polynomial time algorithm that given as input the matrix $A$ and vector $b$ outputs an $m' \times n$ submatrix $A'$ 
of $A$ and vector $b'$ such that $m' \leq (2\alpha+1)^n$  and for every $x \in \mathbb{Z}^n$, $Ax \leq b$ if and only if $A'x \leq b'$.
\end{lemma}

\begin{proof}
Suppose $A$ has more than $(2\alpha+1)^n$ rows. Then, by the pigeon hole principle the system  $Ax \leq b$ has two rows $a_i^Tx \leq b_i$ and $a_j^Tx \leq b_j$ where $i \neq j$ but $a_i = a_j$. Without loss of generality $b_i \leq b_j$, and then any $x \in \mathbb{Z}^n$ such that $a_i^Tx \leq b_i$ satisfies $a_j^Tx \leq b_j$. Thus we can safely remove the inequality $a_j^Tx \leq b_j$ from the system, and the lemma follows.
\end{proof}

%
%

%
%

In the following we will assume that the input is on the form~(\ref{eqn:prob}). We will give an algorithm that runs in time $f(n,m,\alpha,\Delta)\cdot L^{O(1)}$, where $L$ is the length of the bit-representation of the input instance. Since we can reduce the input using Lemma~\ref{lem:rowReduce} first and $\Delta$ is upper bounded in terms of $n$ and $\alpha$ this will yield an FPT algorithm for {\sc Integer Quadratic Programming} parameterized by $n$ and $\alpha$.

Let $r$ be the dimension of the nullspace of $C$. Using Cramer's rule (see~\cite{lay59linear}) we can in polynomial time compute a basis $y_1, \ldots y_r$ for the nullspace of $C$, such that each $y_i$ is an integer vector and $|y_i|_\infty \leq \Delta^2$. We let $Y$ be the $n \times r$ matrix whose columns are the vectors $y_1, \ldots y_r$. We will abuse notation and write $y_i \in Y$ to denote that we chose the $i$'th {\em column} $y_i$ of $Y$. We will say that a feasible solution $x$ is {\em deep} if $x + y_i$ and $x - y_i$ are feasible solutions for all $y_i \in Y$. A feasible solution that is not deep is called {\em shallow}.

\begin{lemma}\label{lem:closeIneq} Let $x$ be a shallow feasible solution to~(\ref{eqn:prob}). Then there exists a row $a_j^T$ of $A$ and integer $b_j'$ such that $a_j^Tx = b_j'$ and $b_j' \in \{b_j - \alpha \cdot n \cdot \Delta^2, \ldots, b_j\}$. Further, $a_j^T$ is linearly independent from the rows of $C$.
\end{lemma}

\begin{proof}
We prove the statement for $y_i \in Y$ such that $x + y_i$ is not a feasible solution to~(\ref{eqn:prob}). Then there exists a row $a_j^T$ of $A$ such that $a_j^T(x + y_i) > b_j$, and thus
$$b_j - a_j^Ty_i < a_j^Tx \leq b_j.$$ 
Thus $a_j^Tx = b_j'$ for $b_j' \in \{b_j - \alpha \cdot n \cdot |y_i|_\infty, b_j\}$. Since $|y_i|_\infty \leq \Delta^2$ we have that  $b_j' \in \{b_j - \alpha \cdot n \cdot \Delta^2, b_j\}$. 

We now show that $a_j^T$ is linearly independent from the rows of $C$. Suppose not, then there exists a coefficient vector $\lambda$ such that $\lambda^TC = a_j^T$. But then 
$$a_j^T(x+y_i) = \lambda^TC(x + y_i) =  \lambda^TCx + \lambda^TCy_i =  \lambda^TCx = a_j^Tx \leq b_j,$$
which contradicts that $a_j^T(x + y_i) > b_j$. We conclude that $a_j^T$ is linearly independent from the rows of $C$. The proof for the case when $x - y_i$ is not a feasible solution to~(\ref{eqn:prob}) is symmetric.
\end{proof}

Lemma~\ref{lem:closeIneq} suggests the following branching strategy: either all optimal solutions are deep or Lemma~\ref{lem:closeIneq} applies to some shallow optimal solution $x^\star$. In the latter case the algorithm can branch on the choice of row $a_j^T$ and $b_j'$ and add the equation $a_j^Tx = b_j'$ to the set of constraints. This decreases the dimension of the nullspace of $C$ by $1$. We are left with handling the case when all optimal solutions are deep. 

\begin{lemma}\label{lem:increaseCondition}
For any pair of vectors $x$, $y \in \mathbb{R}^n$ and symmetric matrix $Q \in \mathbb{R}^{n \times n}$, the following are equivalent. 
\begin{enumerate}\setlength\itemsep{-.7pt}
\item\label{itm:noLamb} $(x+y)^TQ(x+y) \geq x^TQx$ and $(x-y)^TQ(x-y) \geq x^TQx$,
\item\label{itm:absVal} $- y^TQy  \leq  2x^TQy \leq y^TQy.$
\end{enumerate}
\end{lemma}

\begin{proof}
Expanding the inequalities of $(\ref{itm:noLamb})$ yields
$$x^TQx + 2x^TQy + y^TQy \geq x^TQx,$$
$$x^TQx - 2x^TQy + y^TQy \geq x^TQx.$$
Cancelling the $x^TQx$ terms and re-organizing yields $2x^TQy \geq - y^TQy$ and $2x^TQy \leq y^TQy$. Since the left hand side of the inequalities is the same we can combine the two inequalities in a single inequality,
$$- y^TQy  \leq  2x^TQy \leq y^TQy,$$
completing the proof. Note that all the manipulations we did on the inequalities are reversible, thus the above argument does indeed prove equivalence and not only the forward direction $(\ref{itm:noLamb}) \rightarrow  (\ref{itm:absVal})$. 
\end{proof}

Lemma~\ref{lem:increaseCondition} suggests a branching strategy to find a deep solution $x^\star$: pick a vector $y_i$ in $Y$ such that $y_i^TQ$ is linearly independent of the rows of $C$, guess the value $z$ of $2(x^\star)^TQy_i$ and add the linear equation $2(x^\star)^TQy_i = z$ to the set of constraints. In each branch the dimension of the nullspace of $C$ decreases by $1$. Thus we are left with the case that all solutions are deep and there is no $y_i$ in $Y$ such that $y_i^TQ$ is linearly independent of the rows of $C$. We now handle this case.

\begin{lemma}\label{lem:optimalNeighbor}
For any deep optimal solution $x^\star$ of the IQP~(\ref{eqn:prob}) and $y_i \in Y$ such that $y_i^TQ$ is linearly dependent of the rows of $C$, the vectors $x^\star + y_i$ and $x^\star - y_i$ are also optimal solutions. 
\end{lemma}

\begin{proof}
We prove the statement for $x^\star + y_i$. Since $x^\star$ is deep it follows that $x^\star + y_i$ is feasible, and it remains to lower bound the objective function value of $x^\star + y_i$. Since $y_i^TQ$ is linearly dependent of the rows of $C$ there exists a coefficient vector $\lambda^T$ such that $y_i^TQ = \lambda^TC$. Therefore,
\begin{eqnarray*}
2(x^\star + y_i)^TQy_i & = & 2y_i^TQ(x^\star+y_i) = 2\lambda^TC(x^\star+y_i) \\
& = & 2\lambda^TCx^\star+2\lambda^TCy_i = 2\lambda^TCx^\star = 2y_i^TQx^\star = 2(x^\star)^TQy_i
\end{eqnarray*}
Since $x^\star$ is deep, it follows that both $x^\star + y_i$ and $x^\star - y_i$ are feasible and therefore cannot have a higher value of the objective function than $x^\star$. Hence, by Lemma~\ref{lem:increaseCondition} we have that $- y_i^TQy_i  \leq  2(x^\star)^TQy_i \leq y_i^TQy_i$. Since $2(x^\star + y_i)^TQy_i = 2(x^\star)^TQy_i$, we have that
$$- y_i^TQy_i  \leq  2(x^\star + y_i)^TQy_i \leq y_i^TQy_i.$$ 
Hence, Lemma~\ref{lem:increaseCondition}  applied to $(x^\star + y_i)$ implies that 
$$(x^\star)^TQ x^\star  = (x^\star+y_i-y_i)^TQ (x^\star+y_i-y_i) \geq (x^\star+y_i)^TQ(x^\star+y_i).$$
This means that the objective function value of $x^\star+y_i$ is at most that of $x^\star$, hence $x^\star+y_i$ is optimal. The proof for $x^\star - y_i$ is symmetric.
\end{proof}

Lemma~\ref{lem:optimalNeighbor} immediately implies the following corollary.

\begin{corollary}\label{cor:optimalClose} Suppose the IQP~(\ref{eqn:prob}) has an optimal solution, all optimal solutions to~(\ref{eqn:prob}) are deep, and for every $y_i \in Y$, $y_i^TQ$ is linearly dependent of the rows of $C$. Then, for every optimal solution $x^\star$ and integer vector $\lambda \in \mathbb{Z}^r$, $x^\star  + Y\lambda$ is also an optimal solution of~(\ref{eqn:prob}).
\end{corollary}

\begin{proof}
Since $x^\star$ is optimal and deep, and for every $y_i \in Y$, $y_i^TQ$ is linearly dependent of the rows of $C$, it follows from Lemma~\ref{lem:optimalNeighbor} that for every $y_i \in Y$,  $x^\star + y_i$ and $x^\star - y_i$ are also optimal solutions of~(\ref{eqn:prob}). Since all optimal solutions are deep, $x^\star + y_i$ and $x^\star - y_i$ are deep. The statement of the corollary now follows by induction on $|\lambda|_1$.
\end{proof}

We are now ready to state the main structural lemma underlying the algorithm for {\sc Integer Quadratic Programming}.

\begin{lemma}\label{lem:mainStruct}
For any Integer Quadratic Program of the form~(\ref{eqn:prob}) that has an optimal solution and any $x_0$ such that $Cx_0 = d$, there exists an optimal solution $x^\star$ such that at least one of the following three cases holds.
\begin{enumerate}\setlength\itemsep{-.7pt}
\item\label{case:shallow} There exists a row $a_j^T$ of $A$ and integer $b_j' \in \{b_j - \alpha \cdot n \cdot \Delta^2, \ldots, b_j\}$ such that $a_j^Tx^\star = b_j'$, and $a_j^T$ is linearly independent from the rows of $C$.
\item\label{case:deep}  There exists a $y_i \in Y$ such that $y_i^TQ$ is linearly independent of the rows of $C$ and 
$2y_i^TQx^\star = z$ for $z \in \{-n^2\Delta^4\alpha, \ldots, n^2\Delta^4\alpha\}$.
\item\label{case:search} $|x^\star-x_0|_1 \leq \Delta^2 \cdot n$.
\end{enumerate}
\end{lemma}

\begin{proof}
Suppose the integer quadratic program~(\ref{eqn:prob}) has a shallow optimal solution $x^\star$. Then, by Lemma~\ref{lem:closeIneq} case~\ref{case:shallow} applies. In the remainder of the proof we assume that all optimal solutions are deep. Suppose now that there is a $y_i \in Y$ such that $y_i^TQ$ is linearly independent of the rows of $C$. Then, since $x^\star$ is a deep optimal solution, both $x^\star + y_i$ and $x^\star - y_i$ are feasible solutions, so 
$(x^\star+y_i)^TQ(x^\star+y_i) \geq (x^\star)^TQx^\star$ and $(x^\star-y_i)^TQ(x^\star-y_i) \geq (x^\star)^TQx^\star$. Thus, Lemma~\ref{lem:increaseCondition} implies that $2y_i^TQx^\star = z$ for $z \in \{-y_i^TQy_i, \ldots, y_i^TQy_i\}$.
Furthermore, $y_i^TQy_i$ is the sum of $n^2$ terms where each term a product of an element of $y_i$ (and thus at most $\Delta^2$), another element of $y_i$, and an element of $Q$. Thus $z \in \{-n^2\Delta^4\alpha, \ldots, n^2\Delta^4\alpha\}$  and therefore case~\ref{case:deep} applies. 

Finally, suppose that all $y_i \in Y$ are linearly dependent of the rows of $C$. Let $\hat{x}$ be an arbitrarily chosen optimal solution to~(\ref{eqn:prob}). Since $C(x_0-\hat{x}) = 0$ and $Y$ forms a basis for the nullspace of $C$ there is a coefficient vector $\lambda \in \mathbb{R}^r$ such that $x_0 = \hat{x} + Y\lambda$. Define $\tilde{\lambda}$ from $\lambda$ by rounding each entry down to the nearest integer. In other words, for every $i$ we set $\tilde{\lambda}_i = \lfloor \lambda_i \rfloor$. Set $x^\star = \hat{x} + Y\tilde{\lambda}$. By Corollary~\ref{cor:optimalClose} we have that  $x^\star$ is an optimal solution to~(\ref{eqn:prob}). But $|x^\star-x_0|_1 = |Y(\tilde{\lambda}-\lambda)|_1  \leq (\max_i |y_i|_1) \cdot n \leq  \Delta(C)^2 \cdot n$, concluding the proof.
\end{proof}

\begin{theorem}\label{thm:main} There exists an algorithm that given an instance of {\sc Integer Quadratic Programming}, runs in time $f(n,\alpha)L^{O(1)}$, and outputs a vector $x \in \mathbb{Z}^n$. If the input IQP has a feasible solution then $x$ is feasible, and if the input IQP is not unbounded, then $x$ is an optimal solution.
\end{theorem}

\begin{proof}
We assume that input is given on the form~(\ref{eqn:prob}). The algorithm starts by reducing the input system according to Lemma~\ref{lem:rowReduce}. After this preliminary step the number of constraints $m$ in the IQP is upper bounded by $(2\alpha + 1)^n$. We give a recursive algorithm, based on Lemma~\ref{lem:mainStruct}. The algorithm begins by computing in polynomial time a basis $Y = {y_1, \ldots, y_r}$ for the nullspace of $C$, as described in the beginning of Section~\ref{sec:IQPalg}. In particular $Y$ is a matrix of integers, and for every $i$, $|y_i|_\infty \leq \Delta^2$.

If the dimension of the nullspace of $C$ is $0$ the algorithm solves the system $Cx = d$ of linear equations in polynomial time. Let $x^\star$ be the (unique) solution to this linear system. If $x^\star$ is not an integral vector, or $Ax^\star \leq b$ does not hold the algorithm reports that the input IQP has no feasible solution. Otherwise it returns $x^\star$ as the optimum.

If $C$ is not full-dimensional, i.e the dimension of the nullspace of $C$ is at least $1$, the algorithm proceeds as follows. For each row $a_j^T$ of $A$ and integer $b_j' \in \{b_j - \alpha \cdot n \cdot \Delta^2, b_j\}$ such that $a_j^T$ is linearly independent from the rows of $C$, the algorithm calls itself recursively on the same instance, but with the equation  $a_j^Tx = b_j'$ added to the system $Cx = d$. Furthermore, for each $y_i \in Y$ such that $y_i^TQ$ is linearly independent of the rows of $C$ and every integer $z \in \{-n^2\Delta^4\alpha, \ldots, n^2\Delta^4\alpha\}$ the algorithm calls itself recursively on the same instance, but with the equation $2y_i^TQx = z$ added to the system $Cx = d$. Finally the algorithm computes an arbitrary (not necessarily integral) solution $x_0$ of the system $Cx = d$, and checks all (integral) vectors within $\ell_1$ distance at most $\Delta^2 \cdot n$ from $x_0$. The algorithm returns the feasible solution with the smallest objective function value among the ones found in any of the recursive calls, and the search around $x_0$.

In the recursive calls, when we add a linear equation to the system $Cx = d$ we extend the matrix $C$ and vector $d$ to incorporate this equation. The algorithm terminates, as in each recursive call the dimension of the nullspace of $C$ is decreased by $1$. Further, any feasible solution found in any of the recursive calls is feasible for the original system. Thus, if the algorithm reports a solution then it is feasible. To see that the algorithm reports an optimal solution, consider an optimal solution $x^\star$ satisfying the conditions of Lemma~\ref{lem:mainStruct} applied to the quadratic integer program~(\ref{eqn:prob}) and vector $x_0$. Either $x^\star$ will be found in the search around $x_0$, or $x^\star$ satisfies the linear constraint added in at least one of the recursive calls. In the latter case $x^\star$ is an optimal solution to the integer quadratic program of the recursive call, and in this call the algorithm will find a solution with the same objective function value. This concludes the proof of correctness.

We now analyze the running time of the algorithm. First, consider the time it takes to search all integral vectors within $\ell_1$ distance at most  $\Delta^2 \cdot n$ from $x_0$. It is easy to see that there are at most $3^{\Delta^2 \cdot n + n}$ such vectors. For the running time analysis only we will treat this search as at most  $3^{\Delta^2 \cdot n + n}$ recursive calls to instances where the dimension of the nullspace of $C$ is $0$. Then the running time in each recursive call is polynomial, and it is sufficient to upper bound the number of leaves in the recursion tree of the algorithm.

We bound the number of leaves of the recursion tree as a function of $n$ -- the number of variables, $m$ -- the number of rows in $A$, $\alpha$ -- the maximum value of an entry in $A$ or $Q$, $\Delta$ -- the maximum absolute value of the determinant of a square submatrix of $C$, and $r$ -- the dimension of the nullspace of $C$. Notice that the algorithm never changes $Q$ or $A$, and that the number of variables remains the same throughout the execution of the algorithm. Thus $n$, $m$ and $\alpha$ do not change throughout the execution. For a fixed value of $n$, $m$ and $\alpha$, we let $T(r, \Delta)$ be the maximum number of leaves in the recursion tree of the algorithm when called on an instance with the given value of $n$, $m$, $\alpha$, $r$ and $\Delta$.

In each recursive call the algorithm adds a new row to the matrix $C$. Let $C'$ be the new matrix after the addition of this row, $r'$ be the dimension of the nullspace of $C$ and $\Delta'$ be the maximum value of a determinant of a square submatrix of $C'$. Since the new added row is linearly independent of the rows of $C$ it follows that $r' = r - 1$ in each of the recursive calls arising from case~\ref{case:shallow} and case~\ref{case:deep} of Lemma~\ref{lem:mainStruct}. The remaining recursive calls are to leaves of the recursion tree.

When the algorithm explores case~\ref{case:shallow}, it guesses a row $a_j$, for which there are $m$ possibilities, and a value for $b_j'$, for which there are $\alpha \cdot n \cdot \Delta^2$ possibilities. This generates  $m \cdot \alpha \cdot n \cdot \Delta^2$ recursive calls. In each of these recursive calls $a_j$ is the new row of $C'$, and so, by the cofactor expansion of the determinant~\cite{lay59linear}, $\Delta' \leq n\alpha\Delta$.

When the algorithm explores case~\ref{case:deep}, it guesses a vector $y_i \in Y$, and there are at most $n$ possibilities for $y_i$. For each of these possibilities the algorithm guesses a value for $z$, for which there are $2n^2\Delta^4\alpha$ possible choices. Thus this generates $2n^3\Delta^4\alpha$ recursive calls. In each of the recursive calls the algorithm makes a new matrix $C'$ from $C$ by adding the new row $2y_i^TQ$. 
We have that $|y_i|_\infty \leq \Delta^2$. Thus, $|2y_i^TQ|_\infty \leq n \cdot \Delta^2 \cdot \alpha$, and the cofactor expansion of the determinant~\cite{lay59linear} applied to the new row yields $\Delta' \leq n^2 \Delta^3 \cdot \alpha$, where $\Delta'$ is the maximum value of a determinant of a square submatrix of $C'$.  It follows that the number of leaves of the recursion tree is gouverned by the following recurrence.
\begin{align*}
T(r, \Delta) \leq  m \cdot \alpha \cdot n \cdot \Delta^2 \cdot T(r-1,  n\alpha\Delta) + 2n^3 \cdot \Delta^4 \cdot \alpha \cdot T(r-1, n^2 \Delta^3 \alpha) + 3^{(\Delta^2 + 1) \cdot n} \\
\leq \alpha \cdot m \cdot n^3 \cdot \Delta^4 \cdot T(r-1, n^2 \Delta^3 \alpha) + 3^{(\Delta^2 + 1) \cdot n}
\end{align*}
The above recurrence is clearly upper bounded by a function of $n$, $m$, $\Delta$ and $\alpha$. Since $m$ is upper bounded by  $(2\alpha+1)^n$ from Lemma~\ref{lem:rowReduce}, the theorem follows.
\end{proof}

\subsection{Detecting Unbounded IQPs}
Theorem~\ref{thm:main} allows us to solve bounded IQPs, and is sufficient for the application to {\sc Optimal Linear Arrangement}. However, it is somewhat unsatisfactory that the algorithm of Theorem~\ref{thm:main} is unable to detect whether the input IQP is bounded or not. Next we resolve this issue. Towards this, we inspect the algorithm of Theorem~\ref{thm:main}.
For purely notational reasons we will consider the algorithm of Theorem~\ref{thm:main} when run on an instance on the form~(\ref{eqn:prob1}). The first step of the algorithm is to put the the IQP on the form~(\ref{eqn:prob}), and then proceed as described in the proof of Theorem~\ref{thm:main}.

The algorithm is recursive, and the only variables that change from one recursive call to the next are the matrix $C$ and the vector $d$. Furthermore, when making a recursive call, the new matrix $C'$ is computed from $C$ by either adding the row $a_j^T$ to $C$ or adding the row $2y_i^TQ$ to $C$. The vector $y_i^T$ is a vector from the basis $Y$ for the nullspace of $C$. In other words {\em $C'$ depends only on $Q$, $A$, $C$ and $i$, and is independent of $b$ and $d$}. Furthermore, the recursion stops when $C$ has full column rank. Thus, the family ${\cal C}$ of matrices $C$ that the algorithm of Theorem~\ref{thm:main} ever generates depends only on the input matrices $Q$ and $A$ (and not on the input vector $b$).
At this point we remark that the only reason we assumed input was on the form~(\ref{eqn:prob1}) rather than~~(\ref{eqn:prob}) was to avoid the confusing sentence ``{\em Thus, the family ${\cal C}$ of matrices $C$ that the algorithm of Theorem~\ref{thm:main} ever generates depends only on the input matrices $Q$ and $A$, {\bf and $C$},''} where the meaning of the matrix $C$ is overloaded.

Let $\hat{\Delta}$ be the maximum absolute value of the determinant of a square submatrix of a matrix $C \in {\cal C}$ ever generated by the algorithm. In other words, $\hat{\Delta}$ is the maximum value of the variable $\Delta$ throughout the execution of the algorithm. Because $\Delta$ only depends on $C$, it follows that $\hat{\Delta}$ only depends on ${\cal C}$, and therefore $\hat{\Delta}$ is a function of the input matrices $A$ and $Q$.

We now discuss all the different vectors $d$ ever generated by the algorithm. In each recursive call, the algorithm adds a new entry to the vector $d$, this entry is either from the set $\{b_j - \alpha \cdot n \cdot \hat{\Delta}^2, b_j\}$ or from the set $\{-n^2\hat{\Delta}^4\alpha, \ldots, n^2 \hat{\Delta}^4\alpha\}$. Thus, any vector $d$ ever generated by the algorithm is at $\ell_\infty$ distance at most $n^2 \hat{\Delta}^4\alpha$ from some vector whose entries are either $0$ or equal to $b_j$ for some $j \leq m$. Given the vector $b$ and the integer $n$, we define the vector set  ${\cal D}(b,n)$ be the set of all integer vectors in at most $n$ dimensions with entries either $0$ or equal to $b_j$ for some $j \leq m$. Observe that $|{\cal D}(b,n)| \leq (m+1)^{n}$. We have that every vector $d$ ever generated by the algorithm is at $\ell_\infty$ distance at most $n^2 \hat{\Delta}^4\alpha$ from some vector in ${\cal D}(b,n)$.

The algorithm of Theorem~\ref{thm:main} generates potential solutions $x$ to the input IQP by finding a (not necessarily integral) solution $x_0$ to the linear system $Cx = d$, and then lists integral vectors within $\ell_1$-distance at most $\hat{\Delta}^2 \cdot n$ from $x_0$. From Theorem~\ref{thm:main} it follows that if the input IQP is feasible and bounded, then one of the listed vectors $x$ is in fact an optimum solution to the IQP. The above discussion proves the following lemma.

\begin{lemma}\label{lem:prepare}
Given an $n \times n$ integer matrix $Q$, and an $m \times n$ integer matrix $A$, let ${\cal C}$ be the set of matrices $C$ generated by the algorithm of Theorem~\ref{thm:main} when run on the IQP
\begin{eqnarray*}
\mbox{Minimize } x^TQx \\
\mbox{ subject to:~~} Ax \leq 0 \\
x \in \mathbb{Z}^n,
\end{eqnarray*}
and let $\hat{\Delta}$ be the maximum absolute value of the determinant of a square submatrix of a matrix $C \in {\cal C}$. Then, for any $m$-dimensional integer vector $b$ such that the IQP~(\ref{eqn:prob1}) is feasible and bounded, there exists a $C \in {\cal C}$, a vector $d_0 \in {\cal D}(b,n)$, and an integer vector $d$ at $\ell_\infty$ distance at most $n^2 \hat{\Delta}^4\alpha$ from $d_0$ such that the following is satisfied. For any $x_0$ such that $Cx_0 = d$, there exists an integer vector $x^*$ at $\ell_1$ distance at most $\hat{\Delta}^2 \cdot n$ from $x_0$ such that $x^*$ is an optimal solution to the IQP~(\ref{eqn:prob1}).
\end{lemma}

The algorithm in Theorem~\ref{thm:main} only adds a row to the matrix $C$ if this row is linearly independent from the rows of $C$. Hence all matrices in ${\cal C}$ have full row rank, and therefore they have right inverses. Specifically, for each $C \in {\cal C}$, we define $C_{right}^{-1} = C^T(CC^T)^{-1}$. It follows that $C C_{right}^{-1} = I$ and that therefore, $x_0 = C_{right}^{-1}d$ is a solution to the system $Cx = d$ for any vector $d$. Note that $C_{right}^{-1}$ is not necessarly an integer matrix, however Cramer's rule~\cite{lay59linear} shows that $C_{right}^{-1}$ is a matrix with rational entries with common denominator $\det(CC^T)$. This leads to the following lemma.

\begin{lemma}\label{lem:prepare2}
There exists an algorithm that given an $n \times n$ integer matrix $Q$, and an $m \times n$ integer matrix $A$ outputs a set ${\cal C}^{-1}$ of rational matrices, and a set ${\cal V}$ of rational vectors with the following property. For any $m$-dimensional integer vector $b$ such that the IQP~(\ref{eqn:prob1}) defined by $Q$, $A$ and $b$ is feasible and bounded, there exists a matrix $C^{-1}_{right} \in {\cal C}$ a vector $v \in {\cal V}$ and a vector $d_0 \in {\cal D}(b,n)$, such that  
$x^* = C^{-1}_{right}d_0 + v$ is an optimal solution to the IQP~(\ref{eqn:prob1}).
\end{lemma}

\begin{proof}
The algorithm starts by applying the algorithm of Lemma~\ref{lem:prepare} to obtain a set ${\cal C}$ of matrices and the integer $\hat{\Delta}$. The algorithm then computes the set ${\cal C}^{-1}$ of matrices, defined as ${\cal C}^{-1} = \{C_{right}^{-1} ~:~ C \in {\cal C}\}$. 
Next the algorithm constructs the set ${\cal V}$ of vectors as follows. For every matrix $C_{right}^{-1} \in {\cal C}^{-1}$, we have that $C_{right}^{-1} = C^T(CC^T)^{-1}$ for some $C \in {\cal C}$. For every integer vector $v_0$ with $|v_0|_\infty \leq n^2 \hat{\Delta}^4\alpha$, and every {\em rational} vector $v_1$ with $|v_1|_1 \leq \hat{\Delta}^2 \cdot n$ and denominator $\det(CC^T)$ in every entry, the algorithm adds $C_{right}^{-1}v_0 + v_1$ to ${\cal V}$. It remains to prove that ${\cal C}^{-1}$ and ${\cal V}$ satisfy the statement of the lemma.

Let $b$ be an m-dimensional integer vector such that the IQP~(\ref{eqn:prob1}) defined by $Q$, $A$ and $b$ is feasible and bounded. By Lemma~\ref{lem:prepare} we have that there exists $C \in {\cal C}$, a vector $d_0 \in {\cal D}(b,n)$, and a vector $d$ at $\ell_\infty$ distance at most $n^2 \hat{\Delta}^4\alpha$ from $d_0$ such that the following is satisfied. For any $x_0$ such that $Cx_0 = d$, there exists an integer vector $x^*$ at $\ell_1$ distance at most $\hat{\Delta}^2 \cdot n$ from $x_0$ such that $x^*$ is an optimal solution to the IQP~(\ref{eqn:prob1}).

Let $C \in {\cal C}$, $d_0 \in {\cal D}(b,n)$, and $d$ be as guaranteed by Lemma~\ref{lem:prepare}, and set $v_0 = d - d_0$. We have that $|v_0|_\infty \leq n^2 \hat{\Delta}^4\alpha$. Let $C^{-1}_{right}$ be the right inverse of $C$ in ${\cal C}^{-1}$, and let $x_0 = C^{-1}_{right}d = C^{-1}_{right}d_0 + C^{-1}_{right}v_0$. We have that ${\cal C}x_0 = d$. 

Thus, there exists an integer vector $x^*$ at $\ell_1$ distance at most $\hat{\Delta}^2 \cdot n$ from $x_0$ such that $x^*$ is an optimal solution to the IQP~(\ref{eqn:prob1}). Let $v_1 = x^* - x_0$, we have that $|v_1|_1 \leq \hat{\Delta}^2 \cdot n$. Further, $x^*$ is an integer vector, while $x_0 = C^{-1}_{right}d$ is a rational vector whose entries all have denominator $\det(CC^T)$. It follows that all entries of $v_1$ have denominator $\det(CC^T)$. Thus $v = C^{-1}_{right}v_0 + v_1 \in {\cal V}$ and $C^{-1}_{right}d_0 + v$ is an optimal solution to the IQP~(\ref{eqn:prob1}), completing the proof.
\end{proof}

Armed with Lemma~\ref{lem:prepare2} we are ready to prove the main result of this section.

\begin{theorem}\label{thm:main2} There exists an algorithm that given an instance of {\sc Integer Quadratic Programming}, runs in time $f(n,\alpha)L^{O(1)}$, and determines whether the instance is infeasible, feasible and unbounded, or feasible and bounded. If the instance is feasible and bounded the algorithm outputs an optimal solution.
\end{theorem}

\begin{proof}
The algorithm of Theorem~\ref{thm:main} is sufficient to determine whether the input instance is feasible, and to find an optimal solution if the instance is feasible and bounded. Thus, to complete the proof it is sufficient to give an algorithm that determines whether the input IQP is unbounded. We will assume that the input IQP is given on the form~(\ref{eqn:prob1}).

Suppose now that the input IQP is unbounded. For a positive integer $\lambda$, consider adding the linear constraints 
$Ix \leq \lambda{\bf 1}$, and $-Ix \leq \lambda{\bf 1}$ to the IQP. Here {\bf 1} is the $n$-dimensional all-ones vector. In other words, we consider the IQP where the goal is to minimize $x^TQx$, subject to $A'x \leq b'$ where $A'$ is obtained from $A$ by adding $2n$ new rows containing $I$ and $-I$, and $b'$ is obtained from $b$ by adding $2n$ new entries with value $\lambda$. There exists a $\lambda_0$ such that for every $\lambda \geq \lambda_0$ this IQP is feasible. Further, for every $\lambda$ the resulting IQP is bounded. Note that the matrix $A'$ {\em does not depend on $\lambda$.}

We now apply Lemma~\ref{lem:prepare2} on $Q$ and $A'$, and obtain a set ${\cal C}^{-1}$ of rational matrices, and a set ${\cal V}$ of rational vectors. We have that for every $\lambda \geq \lambda_0$, there exists a matrix $C^{-1}_{right} \in {\cal C}$ a $v \in {\cal V}$ and a $d_0 \in {\cal D}(b',n)$, such that  $x^* = C^{-1}_{right}d_0 + v$ is an optimal solution to the IQP defined by $Q$, $A'$ and $b'$. Furthermore, the entries of $b'$ are either entries of $b$, $0$ or equal to $\lambda$. Hence $d_0 = d_b + \lambda d_1$ for $d_b \in {\cal D}(b,n)$ and $d_1 \in {\cal D}({\bf 1},n)$. 
We can conclude that there exists a matrix $C^{-1}_{right} \in {\cal C}$ a vector $v \in {\cal V}$, a vector $d_b \in {\cal D}(b,n)$ and a vector $d_1 \in {\cal D}({\bf 1},n)$, such that  
$$C^{-1}_{right}(d_b + \lambda d_1) + v = \lambda \cdot (C^{-1}_{right} d_1) +  (C^{-1}_{right} d_b + v)$$
is an optimal solution to the IQP defined by $Q$, $A'$ and $b'$. 
Thus, the input IQP is unbounded if and only if there exists a choice for $C^{-1}_{right} \in {\cal C}$, $v \in {\cal V}$, $d_b \in {\cal D}(b,n)$ and $d_1 \in {\cal D}({\bf 1},n)$, such that following {\em univariate} quadratic program with integer variable $\lambda$ is unbounded.
\begin{eqnarray*}
\mbox{Minimize } x^TQx \\
\mbox{ subject to:~~} Ax \leq b \\
x =\lambda \cdot (C^{-1}_{right} d_1) +  (C^{-1}_{right} d_b + v) \\
\lambda \in \mathbb{Z}.
\end{eqnarray*}

Hence, to determine whether the input IQP is unbounded, it is sufficient to iterate over all choices of  $C^{-1}_{right} \in {\cal C}$, $v \in {\cal V}$, $d_b \in {\cal D}(b,n)$ and $d_1 \in {\cal D}({\bf 1},n)$, and determine whether the resulting univariate quadratic program is unbounded. Since the number of such choices is upper bounded by a function of $Q$ and $A$, and univariate (both integer and rational) quadratic programming is trivially decidable, the theorem follows.
\end{proof}




\section{Optimal Linear Arrangement Parameterized by Vertex Cover}
We assume that a vertex cover $C$ of $G$ of size at most $k$ is given as input. The remaining set of vertices $I = V(G) - C$ forms an independent set. Furthermore, $I$ can be partitioned into at most $2^k$ sets as follows: for each subset $S$ of $C$ we define $I_S = \{v \in I ~:~ N(v) = S\}$. For every vertex $v \in I$ we will refer to $N(v)$ as the {\em type} of $v$, clearly there are at most $2^k$ different types.

Let $C = \{c_1, c_2, \ldots, c_k\}$. By trying all $k!$ permutations of $C$ we may assume that the optimal solution $\sigma$ satisfies $\sigma(c_i) < \sigma(c_{i+1})$ for every $1 \leq i \leq k-1$. For every $i$ between $1$ and $k-1$ we define the $i$'th {\em gap of} $\sigma$ to be the set $B_i$ of vertices appearing between $c_i$ and $c_{i+1}$ according to $\sigma$. The $0$'th gap $B_0$ is the set of all vertices appearing before $c_1$, and the $k$'th gap $B_k$ is the set of vertices appearing after $c_k$. We will also refer to $B_i$ as ``gap $i$'' or ``gap number $i$''. For every gap $B_i$ and type $S \subseteq C$ of vertices we denote by $I_S^i$ the set $B_i \cap I_S$ of vertices of type $S$ appearing in gap $i$.

We say that an ordering $\sigma$ is {\em homogenous} if, for every gap $B_i$ and every type $S \subseteq C$ the vertices of $I_S^i$ appear consecutively in $\sigma$. Informally this means that inside the same gap the vertices from different sets $I_{S}$ and $I_{S'}$ ``don't mix''.  Fellows et al.~\cite{FellowsHRS13} show that there always exists an optimal solution that is homegenous.



\begin{lemma}\cite{FellowsHRS13} \label{lem:olahom}
There exists a homogenous optimal linear arrangement of $G$.
\end{lemma}

For every vertex $v$ we define the {\em force} of $v$ with respect to $\sigma$ to be
$$\delta(v) = |\{u \in N(v) ~:~ \sigma(u) > \sigma(v)\}| - |\{u \in N(v) ~:~ \sigma(u) < \sigma(v)\}|.$$
Notice that two vertices of the same type in the same gap have the same force. Fellows et al.~\cite{FellowsHRS13} in the proof of Lemma~\ref{lem:olahom} show that there exists an optimal solution that is homogenous, and where inside every gap, the vertices are ordered from left to right in non-decreasing order by their force. We will call such an ordering solution {\em super-homogenous}. As already noted, the existence of a super-homogenous optimal linear arrangement $\sigma$ follows from the proof of Lemma~\ref{lem:olahom} by Fellows et al.~\cite{FellowsHRS13}.

\begin{lemma}\cite{FellowsHRS13} \label{lem:olasuperhom}
There exists a super-homogenous optimal linear arrangement of $G$.
\end{lemma}

Notice that a super-homogenous linear arrangement $\sigma$ is completely defined (up to swapping positions of vertices of the same type) by specifying for each $i$ and each type $S$ the size $|I_S^i|$. For each gap $i$ and each type $S$ we introduce a variable $x_S^i \in \mathbb{Z}$ representing $|I_S^i|$. Clearly the variables $x_S^i$ need to satisfy 
\begin{align}\label{constr:nonneg}\forall i \leq k, \forall S \subseteq C ~~~~ x_S^i \geq 0\end{align}
and 
\begin{align}\label{constr:sumvert}\forall S \subseteq C ~~~~ \sum_{i=0}^k x_S^i = |I_S|.\end{align}
On the other hand, every assignment to the variables satisfying these (linear) constraints corresponds to a  super-homogenous linear arrangement $\sigma$ of $G$ with $|I_S^i| = x_S^i$ for every type $S$ and gap $i$. 

We now analyze the cost of $\sigma$ as a function of the variables. The goal is to show that $val(\sigma, G)$ is a quadratic function of the variables with coefficients that are bounded from above by a function of $k$. The coefficients of this quadratic function are not integral, but {\em half-integral}, namely integer multiples of $\frac{1}{2}$. The analysis below is somewhat tedious, but quite straightforward. For the analysis it is helpful to re-write $val(\sigma,G)$. For a fixed ordering $\sigma$ of the vertices we say that an edge $uv$ {\em flies over} the vertex $w$ if
$$\min(\sigma(u), \sigma(v)) < \sigma(w) < \max(\sigma(u), \sigma(v)).$$
We define the ``fly over'' relation $\sim$ for edges and vertices, i.e $uv \sim w$ means that $uv$ flies over $w$. Since an edge $uv$ with $\sigma(u) < \sigma(v)$ flies over the $\sigma(v)-\sigma(u)-1$ vertices appearing between  $\sigma(u)$ and $\sigma(v)$ it follows that
$$val(\sigma,G) = |E(G)| + \sum_{uv \in E(G)} \sum_{\substack{w \in V(G) \\ uv \sim w}} 1.$$

We partition the set of edges of $G$ into several subsets as follows. The set $E_C$ is the set of all edges with both endpoints in $C$. For every gap $i$ with $i \in \{0,\ldots,k\}$, every $j \in \{1,\ldots,k\}$ and every $S \subseteq C$ we denote by $E_{i,j}^S$ the set of edges whose one endpoint is in $I_S^i$ and the other is $c_j$. Notice that $|E_{i,j}^S|$ is either equal to $x_S^i$ or to $0$ depending on whether vertices of type $S$ are adjacent to $c_j$ or not. We have that

\begin{align}\label{eqn:objfunc}
val(\sigma,G) = |E(G)| + \sum_{c_ic_j \in E_C} \sum_{\substack{w \in V(G) \\ c_ic_j \sim w}} 1 + \sum_{i,j,S} \sum_{uc_j \in E_{i,j}^S} \sum_{\substack{w \in V(G) \\ uc_j \sim w}} 1.
\end{align}
Further, for each edge $c_ic_j \in E_C$ (with $i < j$) we have that 
$$\sum_{\substack{w \in V(G) \\ c_ic_j \sim w}} 1 = j-i-1 + \sum_{p=i}^{j-1} \sum_{S \subseteq C} x_S^p.$$
In other words, the first double sum of Equation~\ref{eqn:objfunc} is a linear function of the variables. Since $|E_C| \leq {k \choose 2}$ the coefficients of this linear function are integers upper bounded by ${k \choose 2}$.

We now turn to analyzing the second part of Equation~\ref{eqn:objfunc}. We split the triple sum in three parts as follows.
\begin{align}
\nonumber \sum_{i,j,S} \sum_{uc_j \in E_{i,j}^S} \sum_{\substack{w \in V(G) \\ uc_j \sim w}} 1 \\
= \sum_{i,j,S} \left( \sum_{uc_j \in E_{i,j}^S} \sum_{\substack{w \in C \\ uc_j \sim w}} 1 + \sum_{uc_j \in E_{i,j}^S} \sum_{\substack{w \in I_S^i \\ uc_j \sim w}} 1 + \sum_{uc_j \in E_{i,j}^S} \sum_{\substack{w \in I - I_S^i \\ uc_j \sim w}} 1\right)\label{eqn:tripleSumExpand}
\end{align}

For any fixed $i$, $j$ and $S$, any edge $uc_j \in E_{i,j}^S$ the number of vertices $w \in C$ such that $uc_j \sim w$ depends solely on $i$ and $j$. It follows that 
$$ \sum_{uc_j \in E_{i,j}^S} \sum_{\substack{w \in C \\ uc_j \sim w}} 1 = f(i,j) \cdot x_S^i$$
for some function $f$, which is upper bounded by $k$ (since $|C| = k$).

Consider a pair of vertices $u$, $w$ in $I_S^i$ and a vertex $c_j \in C$ such that vertices of $u$'s and $w$'s type are adjacent to $c_j$. Either the edge $uc_j$ flies over $w$ or the edge $wc_j$ flies over $u$, but both of these events never happen simulataneously. Therefore, 
$$\sum_{uc_j \in E_{i,j}^S} \sum_{\substack{w \in I_S^i \\ uc_j \sim w}} 1 = {x_S^i \choose 2} = \frac{(x_S^i)^2}{2} - \frac{x_S^i}{2}$$
In other words, this sum is a quadratic function of the variables with coefficients $\frac{1}{2}$ and $-\frac{1}{2}$. Further, if vertices in $I_S$ are not adjacent to $c_j$ this sum is $0$.

For the last double sum in Equation~\ref{eqn:tripleSumExpand} consider an edge $uc_j \in E_{i,j}^S$ and vertex $v \in I_{S'}^{i'}$ such that $S' \neq S$ or $i' \neq i$. If $uc_j$ flies over $v$ then all the edges in $E_{i,j}^S$ fly over all the vertices in $I_{S'}^{i'}$. Let $g(i,j,S,i',S')$ be a function that returns $1$ if vertices in $I_S$ are adjacent to $c_j$ and all the edges in $E_{i,j}^S$ fly over all the vertices in $I_{S'}^{i'}$. Otherwise $g(i,j,S,i',S')$ returns $0$. It follows that 
$$\sum_{uc_j \in E_{i,j}^S} \sum_{\substack{w \in I - I_S^i \\ uc_j \sim w}} 1 = x_S^i \cdot \sum_{(i',S') \neq (i,S)} g(i,j,S,i',S') x_{S'}^{i'}.$$
In other words, this sum is a quadratic function of the variables with $0$ and $1$ as coefficients. 

The outer sum of Equation~\ref{eqn:tripleSumExpand} goes over all $2^k$ choices for $S$, $k+1$ choices for $i$ and $k$ choices for $j$. Since the sum of quadratic functions is a quadratic function, this concludes the analysis and proves the following lemma.

\begin{lemma}\label{lem:quadObjective}
$val(\sigma, G)$ is a quadratic function of the variables $\{x_S^i\}$ with half-integral coefficients between $-2^kk^2$ and $2^kk^2$. Furthermore, there is a a polynomial time algorithm that given $G$ computes the coefficients.
\end{lemma}

For each permutation $c_1, \ldots, c_k$ of $C$ we can make an integer quadratic program for finding the best super-homegenous solution to {\sc Optimal Linear Arrangement} which places the vertices of $C$ in the order $c_1, \ldots, c_k$ from left to right. The quadratic program has variable set $\{x_S^i\}$ and constraints as in Equations~\ref{constr:nonneg} and~\ref{constr:sumvert}. The objective function is the one given by Lemma~\ref{lem:quadObjective}, but with every coefficient multiplied by $2$. This does not change the set of optimal solutions and makes all the coefficients integral. This quadratic program has at most $2^k \cdot (k+1)$ variables, $2^k \cdot (k+2)$ constraints, and all coefficients are between $-2^{k+1}k^2$ and $2^{k+2}k^2$. Furthermore, since the domain of all variables is bounded the IQP is bounded as well. Thus we can apply Theorem~\ref{thm:main} to solve each such IQP in time $f(k) \cdot n$. This proves the main result of this section.

\begin{theorem} 
{\sc Optimal Linear Arrangement} parameterized by vertex cover is fixed parameter tractable.
\end{theorem}

\section{Conclusions}
We have shown that {\sc Integer Quadratic Programming} is fixed parameter tractable when parameterized by the number $n$ of variables in the IQP and the maximum absolute value $\alpha$ of the coefficients of the objective function and the constraints. 
We used the algorithm for  {\sc Integer Quadratic Programming} to give the first FPT algorithm for {\sc Optimal Linear Arrangement} parameterized by the size of the smallest vertex cover of the input graph.

We hope that this work opens the gates for further research on the parameterized complexity of non-linear and non-convex optimization problems. There are open problems abound. For example, is {\sc Integer Quadratic Programming} fixed parameter tractable when parameterized just by the number of variables? What about the parameterization by $n+m$, the number of variables plus the number of constraints? 
It is also interesting to investigate the parameterized complexity of {\sc Quadratic Programming}, i.e. with real-valued variables rather than integer variables. 
Finally, there is no reason to stop at quadratic functions. In particular, investigating the parameterized complexity of special cases of (integer) mathematical programming with degree-bounded polynomials in the objective function and constraints looks like an exciting research direction. Of course, many of these problems are undecidable~\cite{koppe2012complexity}, but for the questions that are decidable, parameterized complexity might well be the right framework to study efficient algorithms.
\bibliographystyle{siam}
\bibliography{iqp}
\end{document}